\documentclass[11pt,oneside]{article}
\usepackage[top=1in, bottom=1.2in, left=1.1in, right=1.1in]{geometry}
\usepackage{setspace}
\usepackage{amsmath, amsthm, amssymb,amsbsy,amsfonts,mathrsfs}
\usepackage{color}
\usepackage{bm}
\usepackage{natbib}
\usepackage{authblk}
\usepackage{appendix}
\usepackage{float}
\usepackage[caption = false]{subfig}
\usepackage{graphicx}
\usepackage{enumitem}
\usepackage{enumitem}
\usepackage[breaklinks=true]{hyperref}
\usepackage{booktabs}
\usepackage[flushleft]{threeparttable}
\usepackage{titlesec}

\csname @openrightfalse\endcsname

\setlist{nolistsep}
\makeatletter
\newcommand*{\rom}
[1]{\expandafter\@slowromancap\romannumeral #1@}
\makeatother
\DeclareMathOperator*{\argmin}{arg\min}

\newtheorem{thm}{Theorem}

\newtheorem{lemma}{Lemma}

\newtheorem{ass}{Assumption}

\hypersetup{
     colorlinks   = true,
     citecolor    =  blue
}

\newcommand{\Ex}{\mathbb{E}}

\newcommand{\fx}{\mathsf{f}}

\newcommand{\AverageN}{\frac{1}{N}\sum_{i=1}^{N}}
\newcommand{\AverageT}{\frac{1}{T}\sum_{t=1}^{T}}
\newcommand{\AverageNT}{\frac{1}{NT}\sum_{i=1}^{N}\sum_{t=1}^{T}}

\begin{document}

\title{\textbf{A Simple Estimator For Quantile Panel Data Models Using Smoothed Quantile Regressions}}
\author[1]{Liang Chen\thanks{Email: chen.liang@mail.shufe.edu.cn. Financial support from the National Natural Science Foundation
of China (Grant No. 71703089) is gratefully acknowledged.}}
\author[1]{Yulong Huo}
\affil[1]{School of Economics, Shanghai University of Finance and Economics}

\clearpage\maketitle
\thispagestyle{empty}

\begin{abstract}
Canay (2011)'s two-step estimator of quantile panel data models, due to its simple intuition and low computational cost, has been widely used in empirical studies in recent years. In this paper, we revisit the estimator of Canay (2011) and point out that in his asymptotic analysis the bias of his estimator due to the estimation of the fixed effects is mistakenly omitted, and that such omission will lead to invalid inference on the coefficients. To solve this problem, we propose a similar easy-to-implement estimator based on smoothed quantile regressions. The asymptotic distribution of the new estimator is established and the analytical expression of its asymptotic bias is derived. Based on these results, we show how to make asymptotically valid inference based on both analytical and split-panel jackknife bias corrections. Finally, finite sample simulations are used to support our theoretical analysis and to illustrate the importance of bias correction in quantile regressions for panel data.

\vspace{0.5cm}

\end{abstract}

\newpage

\section{Introduction}

Starting with \cite{koenker2004quantile}, the last decade has seen a rapid growth of the literature on quantile regressions for panel data models.  \cite{abrevaya2008effects}, \cite{rosen2012set}, \cite{arellano2016nonlinear} and \cite{graham2018quantile} considered the identification and estimation of quantile effects with fixed $T$\footnote{Throughout the paper we use $N$ and $T$ to denote the number of cross-sectional ovbservations and the number of time-series observations respectively.}; \cite{lamarche2010robust} and \cite{galvao2010penalized} proposed penalized quantile regressions for panel data models with large $T$; \cite{galvao2011quantile} considered quantile regressions of dynamic panels with large $T$; \cite{kato2012asymptotics} and \cite{galvao2016smoothed} focused on the asymptotic distributions of quantile regressions and smoothed quantile regressions; \cite{galvao2013estimation} studied censored quantile regressions for panel data; quantile panel models with interactive fixed effects were considered by \cite{harding2014estimating}, and more recently by \cite{ando2019quantile} and \cite{chen2019twostep}.

Among these methods, the estimation approach of \cite{canay2011simple} is one of the most widely used in empirical studies. According to \cite{10.1093/ectj/utz012}, \cite{canay2011simple} was cited by 120 papers, 81 of which employed its estimator. We refer to \cite{10.1093/ectj/utz012} for an excellent summary of these empirical studies.

The model considered by \cite{canay2011simple} is a standard panel data model with individual fixed effects, where the unobserved idiosyncratic errors are subject to conditional quantile restrictions. The estimation method proposed by \cite{canay2011simple} consists of two steps: in the first step, the individual effects are estimated using the standard fixed effects estimator for linear panel data models; in the second step, the coefficients of the regressors are estimated using standard quantile regressions, treating the estimated individual effects from the first step as given. The intuition behind this two-step estimator is simple, and the its implementation in practice is very easy --- this explains its popularity among empirical researchers.

However, simplicity comes at costs. First, the consistency of \cite{canay2011simple}'s two-step estimator requires certain moments of the idiosyncratic errors to exist, thus the robustness of quantile regressions against heavy-tailed distributions is sacrificed. In comparison, \cite{galvao2016smoothed} estimate the coefficients and the individual effects jointly, and they only require the density functions of the idiosyncratic errors to exist. Second, the validity of \cite{canay2011simple}'s asymptotic analysis was called into question recently by  \cite{10.1093/ectj/utz012}, who discussed two potential errors in \cite{canay2011simple}'s theoretical results: (i) the assumption that $N/T^s\rightarrow 0$ for some $s>1$ is not enough to ignore the asymptotic bias of the estimated coefficients, and (ii) the asymptotic variance of the estimated coefficients derived by \cite{canay2011simple} is not correct.

While the second problem raised by \cite{10.1093/ectj/utz012} can be easily solved by deriving the correct asymptotic variance that takes into account both cross-sectional and serial correlations of the within-transformed regressors (see the exact definition below), the first problem is related to a fundamental issue in fixed effects estimator of nonlinear panel data models. To facilitate the discussion, let $\{Y_{it},X_{it}\}_{1\leq i\leq N,1\leq t\leq T}$ be a panel of observed variables, and let $\{\alpha_i\}_{1\leq i\leq N}$ be the unobserved individual effects. Let $\mathcal{L}$ be a smooth function where the true coefficients is defined by $\beta_0 = \argmin \Ex\left[  \mathcal{L}( Y_{it},\beta'X_{it} +\alpha_i ) | \alpha_i\right]$, so the fixed effects estimator is given by
\[ (\hat{\beta}^{FE},\hat{\alpha}_{1}^{FE},\ldots,\hat{\alpha}^{FE}_N)
=\argmin_{\beta,\alpha_1,\ldots,\alpha_N} \frac{1}{NT}\sum_{i=1}^{N}\sum_{t=1}^{T} \mathcal{L}( Y_{it},\beta'X_{it} +\alpha_i ).
\]
Under some regularity conditions, \cite{hahn2004jackknife} showed that
\[  \hat{\beta}^{FE} -\beta_0 = \frac{1}{\sqrt{NT}}\cdot \mathcal{N}(0,\mathbf{V}) + \frac{b}{T} +o_P(T^{-1}),\]
where $\mathcal{N}(0,\mathbf{V}) $ denotes a vector of normal random variables with means 0 and covariance matrix $\mathbf{V}$, and $b$ is a bias term due to the estimation errors $\hat{\alpha}_{1}^{FE},\ldots,\hat{\alpha}^{FE}_N$. Thus, under the assumption that $N/T\rightarrow \kappa^2>0$, it follows that $ \sqrt{NT} (\hat{\beta}^{FE} -\beta_0 ) \overset{d}{\rightarrow} \mathcal{N}(\kappa  b,\mathbf{V})$. The presence of $\kappa  b$ in the asymptotic distribution of $\hat{\beta}^{FE}$ means that the standard inference on $\beta_0$ using a consistent estimator $\hat{\mathbf{V}}$ of $\mathbf{V}$ is not valid. For example, let $\beta_j$ denote the $j$th element of $\beta$ and let $\mathbf{V}_{jj}$ denote the $j$th diagonal element of $\mathbf{V}$, then the coverage probability of the confidence interval
\[  \left[ \hat{\beta}^{FE}_j -1.96 \sqrt{ \hat{\mathbf{V}}_{jj}/NT  }, \hat{\beta}^{FE}_j +1.96 \sqrt{ \hat{\mathbf{V}}_{jj}/NT  } \right]    \] will not converge to $95\%$ (the nominal level) as $N,T$ go to infinity unless $b=0$. To solve this problem, one can use either analytical bias correction (see \citealt{hahn2004jackknife}) or split-panel jackknife (see \citealt{dhaene2015split}) to alleviate the term $b/T$. Alternatively, in applications where $T$ is much larger than $N$, the asymptotic bias can be simply ignored because $\kappa$ is close to 0. However, for the fixed effects estimator of quantile panel data models where $\mathcal{L}( Y_{it},\beta'X_{it} +\alpha_i ) = \rho_{\tau}(Y_{it} - \beta'X_{it}-\alpha_i)$ and $\rho_{\tau}(u)= (\mathbf{1}\{u\leq 0\}-\tau) u$ is the check function, \cite{kato2012asymptotics} showed that due to the non-smoothness of the check function,
\[  \hat{\beta}^{FE} -\beta_0 \approx \frac{1}{\sqrt{NT}}\cdot \mathcal{N}(0,\mathbf{V}) + O_P\left( \frac{1}{T^{3/4}} \right) ,\]
thus the condition $N^2/T\rightarrow 0$ is required to ignore the asymptotic bias\footnote{In fact, \cite{kato2012asymptotics} explains why it is in general very difficult to derive the analytical expression of the $O_P(T^{-3/4})$ term. Thus, it is unknown whether this term is a bias or a variance term.}. Based on the analysis of \cite{kato2012asymptotics}, we first decompose the stochastic expansion of \cite{canay2011simple}'s two-step estimator, point out an error in \cite{canay2011simple}'s proof that leads to the omission of the asymptotic bias, and argue that $N^2/T\rightarrow 0$ is needed to ignore the asymptotic bias --- this is in contrast to \cite{10.1093/ectj/utz012}'s claim that $N/T\rightarrow 0$ is the required condition for the asymptotic bias to disappear. As discussed above, ignoring the asymptotic bias could result in series problems in the inference on the true coefficients. In one of the simulated model, we find that the coverage rates of the confidence intervals (with $95\%$ nominal level) based on \cite{canay2011simple}'s estimator is lower than $3\%$ when $N=1000$ and $T=20$.

That being said, the main goal of our paper is to provide an alternative easy-to-implement estimator for quantile panel data models. Therefore, our paper can be viewed as both an complement and an extension of \cite{canay2011simple} and \cite{10.1093/ectj/utz012}\footnote{\cite{10.1093/ectj/utz012} also proposed an alternative estimator to reduce the bias, but the asymptotic analysis was not provided.}. The new estimator consists of two steps, where the first step is identical to the first step of \cite{canay2011simple}'s estimation approach, and in the second step we use smoothed quantile regressions instead of standard quantile regressions to estimate the coefficients of the regressors, treating the estimated fixed effects from the first step as given.  Thus, for the many empirical researchers who find \cite{canay2011simple}'s method attractive because of its computational convenience, the cost of learning the new estimator is very low. More importantly, this new estimator allows us to derive the analytical expression of its asymptotic bias. Given the analytical expression of bias, we show that it is fairly easy to constructed bias-corrected estimators and to make valid inference on the true coefficients.

The rest of the paper is organized as follows: In Section 2 we revisit \cite{canay2011simple}'s estimator and point out the main error in his asymptotic analysis. Section 3 introduces an alternative two-step estimator for quantile panel data models, and establishes its asymptotic distribution. We also discuss how to do bias corrections and how to make valid inference based on the bias-corrected estimators. Section 4 provides finite sample simulation results to support our theoretical analysis and to illustrate the importance of bias correction. Finally, Section 5 concludes.

\section{Revisiting Canay (2011)'s Estimator}
\subsection{The Model and Estimator}
Following \cite{canay2011simple},  we consider the following panel data model:
\begin{equation}\label{model1}
Y_{it} =\beta(U_{it})'W_{it}+\alpha_i,
\end{equation}
where $W_{it}=[1,X_{it}']'$, $ X_{it}\in\mathbb{R}^{d}$, $\beta(\cdot): [0,1] \mapsto \mathbb{R}^{d+1}$, $U_{it}\sim \mathcal{U}[0,1]| (X_{it},\alpha_i)$, and $\alpha_i$ represents the time-invariant individual effect. Note that $\beta_1(U_{it})$, the first element of $\beta(U_{it})$, can not be separately identified from $\alpha_i$. Thus, the normalization $\Ex[\beta_1(U_{it}) ]=0$ is imposed throughout the paper. Assuming that the mapping $\tau: \mapsto \beta(\tau)'W_{it}$ is strictly increasing for all $W_{it}$, then we have
\[ P[Y_{it}\leq \beta(\tau)'W_{it}+\alpha_i | X_{it},\alpha_i ] =\tau,\]
or
\begin{equation} \mathsf{Q}_{Y_{it}}[\tau|X_{it},\alpha_i] =\beta(\tau)'W_{it}+\alpha_i  .\end{equation}

Suppose that there is a panel of observed variables $(Y_{it},W_{it})$ for $i=1,\ldots, N$ and $t=1,\ldots,T$. The main object of interest is the partial quantile effect: $\beta(\tau) = \partial \mathsf{Q}_{Y_{it}}[\tau|X_{it},\alpha_{i}] /\partial W_{it}$. The two-step estimator for $\beta(\tau)$ of \cite{canay2011simple} can be defined as follows. First, let
\begin{equation*}
\hat{\alpha}_i = \bar{Y}_i -  \hat{\theta}'\bar{X}_i,
\end{equation*}
where $\bar{Y}_i =T^{-1}\sum_{t=1}^{T}Y_{it}$, $\bar{X}_i =T^{-1}\sum_{t=1}^{T}X_{it}$, and $\hat{\theta}$ is the standard fixed effect estimator, i.e.,
\[     \hat{\theta} = \left( \sum_{i=1}^{N}\sum_{t=1}^{T} \ddot{X}_{it}\ddot{X}_{it}' \right)^{-1}\left( \sum_{i=1}^{N}\sum_{t=1}^{T} \ddot{X}_{it}\ddot{Y}_{it} \right) ,  \]
where $\ddot{X}_{it} =X_{it} -\bar{X}_i$ and $\ddot{Y}_{it} =X_{it} -\bar{Y}_i$ are the within-transformed regressors and dependent variables. In the second step, $\beta(\tau)$ is simply estimated by:
\[ \tilde{\beta}(\tau) =  \argmin_{\beta} \sum_{i=1}^{N}\sum_{t=1}^{T}\rho_{\tau} \left(  Y_{it} -\beta'W_{it}-\hat{\alpha}_i \right).\]

\subsection{Expansion of $ \tilde{\beta}(\tau) -\beta(\tau)$}

To simply the notations, let $\{\alpha_{01}, \ldots, \alpha_{0N}\} $ be the realized values of the individual effects, and our analysis in this paper are conditional on them.

Note that by defining $\lambda_0 =[0,\theta_0']' = \Ex[ \beta(U_{it})]$, model \eqref{model1} can be written as (conditional on $\alpha_{01},\ldots,\alpha_{0N}$)
\[
  Y_{it} =\lambda_0'W_{it}+\alpha_{0i} +\epsilon_{it} =\theta_0'X_{it}+\alpha_{0i}+\epsilon_{it},\text{ where }  \epsilon_{it} = \left( \beta(U_{it}) -\lambda_0\right) 'W_{it}, \]
and $\Ex[\epsilon_{it}|X_{it}]=0$. In other words, model \eqref{model1} can be transformed into a standard linear panel data model where $\alpha_{01},\ldots,\alpha_{0N}$ can be consistently estimated using the fixed effects estimator in the first step.


Define $\psi(u) = \tau - \mathbf{1}\{u\leq 0\}$, and expanding $\Ex[ \psi( Y_{it} -\tilde{\beta}(\tau)'W_{it}-\hat{\alpha}_i )W_{it}  ]$ around $(\beta(\tau),\alpha_{0i})$ gives
\begin{multline}\label{eq3}
-\Ex[ \psi( Y_{it} -\tilde{\beta}(\tau)'W_{it}-\hat{\alpha}_i )W_{it}  ] =  \Ex[ \mathsf{f}_{it}(0|X_{it})W_{it} W_{it}' ]\cdot(\tilde{\beta}(\tau) -\beta(\tau) )
+\Ex[ \mathsf{f}_{it}(0|X_{it})W_{it} ]\cdot(\hat{\alpha}_i -\alpha_{0i} ) \\
- 0.5\cdot \Ex[ \mathsf{f}^{(1)}_{it}(0|X_{it})W_{it} ]\cdot(\hat{\alpha}_i -\alpha_{0i} )^2 + O(1)\cdot (\hat{\alpha}_i -\alpha_{0i} )^3 +o_P(\|\tilde{\beta}(\tau) -\beta(\tau) \|),
\end{multline}
where
\[u_{it} =  \left( \beta(U_{it}) -\beta(\tau)\right) 'W_{it} =Y_{it} - \beta(\tau)'W_{it}-\alpha_{0i}, \] $\fx_{it}$ is the density function of $u_{it}$, and $\fx_{it}^{(j)}(u) =\partial^j \fx(u)/\partial u^j $. Next, assume stationarity and define
\[\Sigma = \lim_{N\rightarrow\infty} \AverageN   \Ex \left[ \fx_{it}(0|X_{it}) W_{it}W_{it}' \right], \quad \gamma_i=\Ex[ \mathsf{f}_{it}(0|X_{it})W_{it} ],
\quad \eta_i = \Ex[ \mathsf{f}^{(1)}_{it}(0|X_{it})W_{it} ],
\]
\[  \mathbb{V}_{NT} (\beta,\alpha_1,\ldots,\alpha_N) =\frac{1}{\sqrt{NT}}\sum_{i=1}^{N}\sum_{t=1}^{T}\{ \psi( Y_{it} -\beta'W_{it}-\alpha_i )W_{it} -
\Ex[ \psi( Y_{it} -\beta'W_{it}-\alpha_i )W_{it}]\}.
\]
By the computational properties of the quantile regressions and equation \eqref{eq3}, we have the following stochastic expansion for $\tilde{\beta}(\tau) -\beta(\tau)$:
\begin{multline}\label{eq4}
\Sigma \cdot(\tilde{\beta}(\tau) -\beta(\tau) ) =\frac{1}{NT} \sum_{i=1}^{N}\sum_{t=1}^{T}\psi( u_{it} )W_{it}-\frac{1}{N} \sum_{i=1}^{N}\gamma_i\cdot(\hat{\alpha}_i -\alpha_{0i} ) + 0.5\frac{1}{N} \sum_{i=1}^{N}\eta_i\cdot(\hat{\alpha}_i -\alpha_{0i} )^2\\
+ \frac{1}{\sqrt{NT}} \left(\mathbb{V}_{NT} (\tilde{\beta}(\tau),\hat{\alpha}_{1},\ldots,\hat{\alpha}_{N})  -  \mathbb{V}_{NT} (\beta(\tau),\alpha_{01},\ldots,\alpha_{0N})  \right)
+o_P(\|\tilde{\beta}(\tau) -\beta(\tau) \|) +o_P(T^{-1}).
\end{multline}
Similar to the proof Theorem 2 below, it can be shown that
\[ \frac{1}{NT} \sum_{i=1}^{N}\sum_{t=1}^{T}\psi( u_{it} )W_{it}-\frac{1}{N} \sum_{i=1}^{N}\gamma_i\cdot(\hat{\alpha}_i -\alpha_{0i} ) + 0.5\frac{1}{N} \sum_{i=1}^{N}\eta_i\cdot(\hat{\alpha}_i -\alpha_{0i} )^2= \frac{1}{\sqrt{NT}}\cdot \mathcal{N}(0, \Omega) + \frac{c}{T} + o_P(T^{-1}) \]
where $\Omega$ is a covariance matrix which will be defined in Theorem 2 (also see Remark 4 for a mistake in \citealt{canay2011simple}'s expression for the variance matrix) and $c$ is a nonzero constant vector. The key step in \cite{canay2011simple}'s analysis is to show that
\begin{equation}\label{eq5} \|\mathbb{V}_{NT} (\tilde{\beta}(\tau),\hat{\alpha}_{1},\ldots,\hat{\alpha}_{N}) -  \mathbb{V}_{NT} (\beta(\tau),\alpha_{01},\ldots,\alpha_{0N})\|=o_P(1),\end{equation}
which was proved in Lemma A.1 of \cite{canay2011simple}. Inspecting the proof of the above result in \cite{canay2011simple}, it is clear that the following inequality was assumed to be true:
\begin{multline}\label{eq6}\|\mathbb{V}_{NT} (\tilde{\beta}(\tau),\hat{\alpha}_{1},\ldots,\hat{\alpha}_{N}) -  \mathbb{V}_{NT} (\beta(\tau),\alpha_{01},\ldots,\alpha_{0N})\|  \\
\leq  \sup_{\mbox{\tiny$\begin{array}{c}
\| \beta-\beta(\tau)\| \leq \| \tilde{\beta}(\tau) -\beta(\tau) \| \\
\| \alpha_a -\alpha_b\| \leq \max_{1\leq i\leq N}\| \hat{\alpha}_i -\alpha_{0i}\| \end{array}$}}  \| \mathbb{U}_{NT} (\beta,\alpha_{a}) -  \mathbb{U}_{NT} (\beta(\tau),\alpha_b)\|,
 \end{multline}
where
\[\mathbb{U}_{NT} (\beta,\alpha) =\frac{1}{\sqrt{NT}}\sum_{i=1}^{N}\sum_{t=1}^{T}\{ \psi( Y_{it} -\beta'W_{it}-\alpha )W_{it} -
\Ex[ \psi( Y_{it} -\beta'W_{it}-\alpha )W_{it}]. \]
By the consistency of $\tilde{\beta}$ and the uniform consistency of $\hat{\alpha}_i$, the right-hand side of \eqref{eq6} can be shown to be $o_P(1)$ because the empirical process $\mathbb{U}_{NT}$ is stochastically equicontinous.

However, it not difficult to see that inequality \eqref{eq6} does not hold in general, and thus the proof of \eqref{eq5} in \cite{canay2011simple} is not correct. In fact, using the arguments of \cite{kato2012asymptotics}, one can show that
\[ \frac{1}{\sqrt{NT}} \left( \mathbb{V}_{NT} (\beta(\tau),\alpha_{01},\ldots,\alpha_{0N})  -\mathbb{V}_{NT} (\tilde{\beta}(\tau),\hat{\alpha}_{1},\ldots,\hat{\alpha}_{N}) \right) \approx O_P\left( \frac{1}{T^{3/4}} \right).\]
Therefore, under the assumption that $N/T\rightarrow \kappa^2>0$, we have
\[ \sqrt{NT}(  \tilde{\beta}(\tau) -\beta(\tau)) \approx \mathcal{N}(\kappa \cdot   \Sigma^{-1} c , \Sigma^{-1}\Omega\Sigma^{-1} ) +O_P\left( \frac{\sqrt{N}}{ T^{1/4}}\right) +o_P(1). \]
Thus, similar to \cite{kato2012asymptotics}, the condition on $N,T$ to ignore the asymptotic bias of the estimator is that $N^2/T\rightarrow 0$, which is different from \cite{canay2011simple}'s assumption that $N/T^s\rightarrow 0$ for some $s>1$ and \cite{10.1093/ectj/utz012}'s claim that $N/T\rightarrow 0$ is required.

Moreover, even if \eqref{eq6} is right and the $O_P(\sqrt{N}/T^{1/4})$ term can be dropped from the above equation, \cite{canay2011simple} still made two mistakes in deriving the asymptotic distribution of his estimator: the asymptotic bias $\kappa \cdot   \Sigma^{-1} c$ is omitted and the expression of $\Omega$ is not correct (see Remark 4 below).

The consequence of ignoring the asymptotic bias of \cite{canay2011simple}'s estimator when $T$ is small compared to $N$ is illustrated using Monte Carlo simulations in Section 4, where we show that the common practice of constructing confidence intervals using \cite{canay2011simple}'s estimator could result in coverage rates that are much lower than the nominal level.


\section{A New Estimator Based on Smoothed Quantile Regressions}
\subsection{The New Estimator}
In this paper, to solve \cite{canay2011simple}'s problem discussed in the previous section, we propose a new two-step estimator based on smoothed quantile regression (SQR hereafter). The first step of our estimation method is the same as the first step of \cite{canay2011simple}'s two-step estimator, i.e., the individual effects are estimated using the standard fixed effects estimators for linear panel data models: $\hat{\alpha}_i =  \bar{Y}_i -  \hat{\theta}'\bar{X}_i$, where $\hat{\theta}$ is defined in Section 2.
%

In the second step, inspired by \cite{galvao2016smoothed}, we propose to estimate $\beta(\tau)$ using the following SQR:
\begin{equation} \hat{\beta}(\tau) =  \argmin_{\beta} \sum_{i=1}^{N}\sum_{t=1}^{T} \left[  \tau - K\left( \frac{Y_{it} -\beta'W_{it}-\hat{\alpha}_i}{h}\right)  \right] \cdot \left(  Y_{it} -\beta'W_{it}-\hat{\alpha}_i \right) ,\end{equation}
where $K(z) =1- \int_{-\infty}^{z}k(u)du $, $k(\cdot)$ is a continuous function with support $[-1,1]$ and symmetric around 0, and $h$ is a bandwidth parameter.

\vspace{0.2cm}
\noindent{\textbf{Remark 1}: }As pointed out by \cite{kato2012asymptotics}, the main difficulty in deriving the analytical expression for the bias of \cite{canay2011simple}'s estimator originates from the non-smoothness of the check function. The main motivation of using SQR in the second step of the new estimator is to approximate the indicator function by a smooth function. Similar ideas has been explored by \cite{amemiya1982two} and \cite{horowitz1998bootstrap}, but for different objectives. The main purpose of using SQR in our estimator is that it allows us to work out the analytical expression of the asymptotic bias of the estimator, which provides the theoretical basis of using analytical and split-panel jackknife bias corrections.  \qed

\vspace{0.2cm}
\noindent{\textbf{Remark 2}: }In terms of computational cost, the new estimator is slightly more complicated than the estimator of \cite{canay2011simple}, because in the second step the new estimator has to solve a nonlinear minimization problem, while the standard quantile regression in the second step of \cite{canay2011simple}'s estimator can be efficiently solved by linear programming. However, since it only estimate $d+1$ parameters in the second step, the new estimator is still much simpler than the estimator of \cite{galvao2016smoothed} which estimate $d+1+N$ parameters in a nonlinear minimization problem. Moreover, to reduce the computational cost, we can use \cite{canay2011simple}'s estimator (which is consistent) as the initial value in our second step.  \qed

%

\subsection{Consistency}
To prove the consistency of the new estimator, we impose the following conditions:

\begin{ass}  Let $C, \underline{\rho}$ be positive constants and $\mu_i=\Ex[X_{it}]$, \\
(i) $(u_{it},X_{it})$ is independent of $(u_{js},X_{js})$ for any $i\neq j$. For each $i$, the distributions of $(u_{i1},X_{i1})$,\ldots, $(u_{iT},X_{iT})$ are identical.\\
(ii) $\rho_{\text{min}}  \left( \Ex \left[ \fx_{it}(0|X_{it}) W_{it}W_{it}' \right]  \right) \geq \underline{\rho}$ for all $i,t$.    \\
(iii) $\Ex \left[  \|X_{it}\|^4\right]\leq C$, and $\Ex[|u_{it}|^4 ]\leq C$ for all $i,t$.\\
(iv) $h\rightarrow0$ as $N,T\rightarrow\infty$.\\
(v) $ \sup_{\tau\in(0,1)} \| \beta(\tau) \| \leq C $. \\
(vi) Let $\tilde{X}_{it}=X_{it} -\mu_i$. Then $\mathbf{B} =   \lim_{N\rightarrow\infty}N^{-1}\sum_{i=1}^{N} \Ex[\tilde{X}_{it}\tilde{X}_{it}']$
is positive definite.
\end{ass}
Assumption 1(i), which is also imposed by \cite{canay2011simple}, is admittedly strong, but it can be relaxed at the expense of much lengthier proofs to allow for serially dependence such as $\beta$-mixing. Assumption 1(ii) is the standard identification condition in quantile regressions, and it is widely used in the literature.

 Unlike \cite{kato2012asymptotics} and \cite{galvao2016smoothed} that only require the density of $u_{it}$ to exist, the consistency our estimator needs the fourth moments of $u_{it}$ to be finite. The lost of robustness against heavy-tailed distributions (such as Cauchy distribution) is the price one has to pay for employing the simpler two-step approaches. Note that by definition, $\epsilon_{it}$ is related to $u_{it}$ by
\begin{equation}\label{eq8}
\epsilon_{it} =\left( \beta(U_{it}) -\lambda_0\right) 'W_{it}= u_{it} + [\beta(\tau) -\lambda_0]'W_{it}.
\end{equation}
Thus, Assumptions 1(i) and 1(iii) imply that $\Ex[ |\epsilon_{it}|^4] < \infty$ and that $\epsilon_{it}$ is independent of $\epsilon_{js}$ for any $i\neq j $ or $t\neq s$.

Last but not least, for consistency, $h$ is required to converge to 0 in the limit, but different from \cite{kato2012asymptotics} and \cite{galvao2016smoothed}, we don't impose any restrictions on the relative sizes of $N$ and $T$ as long as they both diverge to infinity.

Then, we can show that:
\begin{thm}Suppose that Assumption 1 holds, then $\|\hat{\beta}(\tau) -\beta(\tau)\| =o_P(1)$.
\end{thm}

\vspace{0.2cm}
\noindent{\textbf{Remark 3}: }Assumption 4.2 of \cite{canay2011simple} assumes that
\begin{equation}\label{eq9}  \sqrt{NT}( \hat{\theta} -\theta_0) = \frac{1}{\sqrt{NT}} \sum_{i=1}^{N}\sum_{t=1}^{T} \delta_{it}  +o_P(1), \end{equation}
where $\delta_{it}$ is an i.i.d sequence of zero-mean random vectors such that $\sqrt{NT} (  \hat{\theta} -\theta_0 ) \overset{d}{\rightarrow}\mathcal{N}(0,\mathbf{V})$ with
\[ \mathbf{V} = \lim_{N,T\rightarrow\infty}\frac{1}{NT} \sum_{i=1}^{N}\sum_{t=1}^{T} \Ex[ \delta_{it}\delta_{it}' ]. \]
However, as pointed out by \cite{10.1093/ectj/utz012}, \eqref{eq9} is unlikely to hold. Instead, in Lemma 1 of the Appendix, we provide a rigorous proof that
\[ \sqrt{NT}( \hat{\theta} -\theta_0) = \mathcal{N}(0, \mathbf{B}^{-1} \Sigma_\theta \mathbf{B}^{-1}) +o_P\left( \sqrt{N}/\sqrt{T}\right),\]
where
\[\Sigma_{\theta} = \lim_{N\rightarrow\infty} \frac{1}{N}\sum_{i=1}^{N} \Ex[ \epsilon_{it}^2 \tilde{X}_{it}\tilde{X}_{it}'] .\]
\qed

\subsection{Asymptotic Distribution}
To establish the asymptotic distribution of the new estimator, the following conditions are imposed:
\begin{ass}(i) $X_{it}\in\mathcal{X}$ for all $i,t$ and $\mathcal{X}$ is compact.\\
(ii) Let $q\geq 4$, and let $\fx_{it}^{(j)}(c) = \partial \fx_{it}^{(j)}(c)/\partial c$, $\fx_{it}^{(j)}(c|X_{it}) = \partial \fx_{it}^{(j)}(c|X_{it})/\partial c$ $j=1,\ldots,q$. Then for each $j$, $|\fx_{it}^{(j)}(c) |$ and $ |\fx_{it}^{(j)}(c|X_{it})|$  are uniformly bounded for all $i,t$.\\
(iii) For $q\geq 4$, $\int_{-1}^{1}k(u)du=1$, $\int_{-1}^{1}k(u)u^jdu=0$ for $j=1,\ldots,q-1$ and $\int_{-1}^{1}k(u)u^qdu \neq 0$.\\
(iv) $N/T \rightarrow \kappa^2>0$ as $N,T\rightarrow \infty$. $h\asymp T^{-c}$ and $1/q < c< 1/3$.
\end{ass}
The above assumptions are identical to Assumptions (A2), (A5), (A6) and (A7) of \cite{galvao2016smoothed}. We refer to \cite{galvao2016smoothed} for the details of these assumptions. The following theorem gives the asymptotic distribution of the new estimator.

%

\begin{thm}Under Assumptions 1 and 2, as $N,T\rightarrow\infty$,
\[ \sqrt{NT}(\hat{\beta}(\tau) -\beta(\tau)) \overset{d}{\rightarrow} \mathcal{N}\left( \kappa b, \Sigma^{-1}\Omega\Sigma^{-1} \right),\]
where
\[b= [\lambda_0 -\beta(\tau) ]+ 0.5 \Sigma^{-1}  \cdot   \lim_{N\rightarrow\infty} N^{-1}\sum_{i=1}^{N}\eta_i \Ex[\epsilon_{it}^2] ,\]
\begin{multline*}
\Omega = \tau(1-\tau)\cdot \lim_{N\rightarrow\infty}\frac{1}{N}\sum_{i=1}^{N} \Ex[  W_{it}W_{it}'] + \lim_{N\rightarrow\infty}\frac{1}{N}\sum_{i=1}^{N}\gamma_i \gamma_i' \cdot \Ex\left[\epsilon_{it}^2 \right] +\mathbf{A}\mathbf{B}^{-1} \mathbf{A}'\\
-2 \lim_{N\rightarrow\infty}\frac{1}{N}\sum_{i=1}^{N} \Ex\left[ \left(\tau - \mathbf{1}\{u_{it}\leq 0\}\right)u_{it}W_{it}      \cdot ( \gamma_i - \mathbf{A}\mathbf{B}^{-1}\tilde{X}_{it})' \right]+2 \lim_{N\rightarrow\infty}\frac{1}{N}\sum_{i=1}^{N}\gamma_i\Ex\left[\tilde{X}_{it}'\epsilon_{it}^2 \right]\mathbf{B}^{-1} \mathbf{A}',
\end{multline*}
and
\[\mathbf{A} = \lim_{N\rightarrow\infty}N^{-1}\sum_{i=1}^{N}\gamma_i \mu_i'.\]
\end{thm}

\vspace{0.2cm}
\noindent{\textbf{Remark 4}: }In the proof of Theorem 2, we establish the following Bahadur representation for $\hat{\beta}(\tau)$:
\begin{equation*}
\sqrt{NT}(\hat{\beta}(\tau) -\beta(\tau)) =\Sigma^{-1} \frac{1}{\sqrt{NT}} \sum_{i=1}^{N} \sum_{t=1}^{T}Z_{it}+\sqrt{\frac{N}{T}}\cdot b+o_P(1) ,
\end{equation*}
where $Z_{it}=\varrho^{(1)}_{it}W_{it}-\gamma_i\epsilon_{it}-\mathbf{A}\mathbf{B}^{-1} \tilde{X}_{it}\epsilon_{it}$,
and $\Omega$ is is limit of $ N^{-1}\sum_{i=1}^{N}\Ex[Z_{it}Z_{it}']$. The term $\mathbf{A}\mathbf{B}^{-1} \tilde{X}_{it}\epsilon_{it}$ represents the effects of estimating $\theta_0$ using the fixed effects estimator in the first step. However, this term is omitted in the covariance matrix derived in Theorem 4.2 of \cite{canay2011simple}.  \qed

\subsection{Bias Correction and Inference}
Theorem 2 provides the basis of analytical and split-panel jackknife bias correction.

First, consider analytical bias correction. Define
\[ \hat{\epsilon}_{it} = Y_{it} -\hat{\theta}'X_{it} -\hat{\alpha}_i, \quad \hat{u}_{it} = Y_{it} -\hat{\beta}(\tau)'W_{it} -\hat{\alpha}_i, \]
\[  \hat{\varrho}^{(1)}_{it} = \tau - K(\hat{u}_{it}/h) ,\quad \hat{\varrho}^{(2)}_{it} = k(\hat{u}_{it}/h)1/h, \quad  \hat{\varrho}^{(3)}_{it} = k^{(1)}(\hat{u}_{it}/h)1/h^2, \]
\[ \hat{\Sigma} =  \AverageNT   \hat{\varrho}^{(2)}_{it} W_{it}W_{it}', \quad \hat{\eta}_i = \frac{1}{T}\sum_{t=1}^{T} \hat{\varrho}_{it}^{(3)}W_{it} ,\]
\[\hat{b} = \hat{\lambda} -\hat{\beta}(\tau)+0.5\cdot \hat{\Sigma}^{-1} \cdot \AverageNT \hat{\eta}_i \hat{\epsilon}_{it}^2, \text{ where } \hat{\lambda}=[0,\hat{\theta}']'.\]
Then the estimator with analytical bias correction is defined as
\[ \hat{\beta}_{abc}(\tau) =\hat{\beta}(\tau) - \hat{b}/T. \]

Next, consider split-panel jackknife method. Let $\hat{\beta}_{1}(\tau)$ be our two-step estimator using the sample $\{ (Y_{it},X_{it}),i=1,\ldots,N,t=1,\ldots,T/2\}$, and let $\hat{\beta}_{2}(\tau)$ be our two-step estimator using the sample $\{ (Y_{it},X_{it}),i=1,\ldots,N,t=1+T/2,\ldots,T\}$. Then the estimator with split-panel jackknife is defined as
\[ \hat{\beta}_{spj}(\tau) =2 \hat{\beta}(\tau) - 0.5( \hat{\beta}_1(\tau)+\hat{\beta}_2(\tau)) . \]

Moreover, to make inference we need to estimate the covariance matrix. Define
\[\hat{\gamma}_i = \frac{1}{T}\sum_{t=1}^{T} \hat{\varrho}_{it}^{(2)}W_{it}, \quad \hat{\mathbf{A}} =\AverageN \hat{\gamma}_i \bar{X}_i', \quad
\hat{\mathbf{B}}=\AverageNT \ddot{X}_{it} \ddot{X}_{it}',
\]
\[ \hat{Z}_{it} = \hat{\varrho}^{(1)}_{it}W_{it} - \hat{\gamma}_i \hat{\epsilon}_{it} -\hat{\mathbf{A}} \hat{\mathbf{B}}^{-1}\ddot{X}_{it}\hat{\epsilon}_{it}.\]
According to Remark 4, the estimator of $\Omega$ can be constructed as
\begin{equation}\label{eq10} \hat{\Omega} = \AverageNT \hat{Z}_{it}\hat{Z}_{it}'.\end{equation}

Under Assumptions 1 and 2, similar to the proof of Theorem 2, it can be shown that
\begin{equation*}
\sqrt{NT}(\hat{\beta}_{abc}(\tau) -\beta(\tau)) \overset{d}{\rightarrow} \mathcal{N}\left(0, \Sigma^{-1}\Omega\Sigma^{-1} \right),
\end{equation*}
\begin{equation*}
\sqrt{NT}(\hat{\beta}_{spj}(\tau) -\beta(\tau)) \overset{d}{\rightarrow} \mathcal{N}\left( 0,\Sigma^{-1}\Omega\Sigma^{-1} \right),
\end{equation*}
and
\begin{equation*}
\hat{\Sigma}^{-1}\hat{\Omega}\hat{\Sigma}^{-1} \overset{p}{\rightarrow}\Sigma^{-1}\Omega\Sigma^{-1}.
\end{equation*}
The above results ensure that inferences based on the bias-corrected estimators and the estimated variance are asymptotically valid. The finite sample performances of $\tilde{\beta}$ (Canay's two-step estimator), $\hat{\beta}_{abc}$, $\hat{\beta}_{spj}$, and the coverage rates of the corresponding confidence intervals are evaluated in the next section.

\section{Finte Sample Simulations}
In this section, we use Monte Carlo simulations to study the finite sample performances of the proposed estimators. To facilitate the comparison, we use the following data generating process (DGP) identical to the ones used by \cite{canay2011simple}:
\[ Y_{it}=(\epsilon_{it}-1)+\epsilon_{it}X_{it}+\alpha_i,\quad  \alpha_i=\gamma(X_{i1}+\dots+X_{iT}+\lambda_i)-\Ex(\alpha_i) ,\]
where $X_{it}\sim i.i.d \text{ }Beta(1,1)$, $\lambda_i\sim i.i.d \text{ }N(0,1)$, and $\gamma=2$. As in \cite{canay2011simple}, we consider three different models with different distributions for $\epsilon_{it}$: in Model 1, $\epsilon_{it}\sim i.i.d \text{ }N(2,1)$; in Model 2, $\epsilon_{it}\sim i.i.d \text{ }\exp (1)+2$; in Model 3, $\epsilon_{it} \sim  i.i.d \text{ }B_{it}\cdot \mathcal{N}(1,1) +(1-B_{it})\cdot\mathcal{N}(3,1)$ with $B_{it} \sim Bernoulli(0.3)$. In addition, to see how the proposed estimators perform when the errors have heavy-tailed distributions, we consider Model 4 where $\epsilon_{it}\sim i.i.d \text{ }\mathcal{T}(5)$ where $\mathcal{T}(5)$ represents Student's t distribution with five degrees of freedom.

We focus on the coefficient of $X_{it}$, which is given by $\mathsf{Q}_{\epsilon}(\tau) $ at quantile $\tau$, where $\mathsf{Q}_{\epsilon}(\tau) $ is the quantile function of $\epsilon_{it}$. The following three estimators of $\beta(\tau)$ are considered:
\begin{itemize}
\item $\tilde{\beta}(\tau):$ Canay's two-step estimator;
\item $\hat{\beta}_{abc}(\tau):$ The new two-step estimator with analytical bias correction;
\item $\hat{\beta}_{spj}(\tau):$ The new two-step estimator with split-panel jackknife bias correction.
\end{itemize}
The biases, mean square errors (MSEs), and the coverage rates of the $95\%$ confidence intervals of the three estimators are compared. To construct the confidence intervals, \eqref{eq10} is used to calculate the variances of the three estimators.

In the SQR, the following fourth-order kernel function is used:
\[ k(u)=\frac{105}{64}(1-5u^2+7u^4-3u^6)\mathbf{1}(|u|\leq 1) \]
and the bandwidth is set as $h=0.8$. We have also tried other choices of $h$ and the results are similar. More simulation results with other choices of $h$ are available upon request.

The simulation results (from 1000 replications) for Model 1 to Model 4 at $\tau=0.25$ and $\tau=0.9$ are reported in Table 1 to Table 4 respectively.

{\small
 \begin{center}
  \begin{threeparttable}
  \caption{Biases, MSEs and Coverage Rates for Model 1}

    \begin{tabular}{ccccccccccc}
    \toprule
   &  &
    \multicolumn{3}{c}{Biases}&\multicolumn{3}{c}{MSEs}&\multicolumn{3}{c}{Coverage Rates ($95\%$)}\cr
    \cmidrule(lr){3-5} \cmidrule(lr){6-8}  \cmidrule(lr){9-11}
   & $(N,T)$ &$\tilde{\beta}(\tau)$&$\hat{\beta}_{abc}(\tau)$&$\hat{\beta}_{spj}(\tau)$
    &$\tilde{\beta}(\tau)$&$\hat{\beta}_{abc}(\tau)$&$\hat{\beta}_{spj}(\tau)$
    &$\tilde{\beta}(\tau)$&$\hat{\beta}_{abc}(\tau)$&$\hat{\beta}_{spj}(\tau)$\cr
    \midrule
  $\tau =0.25$ & (100,10)&0.077&-0.019&0.013&0.061&0.065&0.110&0.895&0.898&0.807\cr
   & (100,20)&0.036&-0.011&-0.001&0.026&0.028&0.044&0.921&0.922&0.832\cr
    &(200,10)&0.076&-0.011&0.006&0.032&0.029&0.048&0.905&0.931&0.838\cr
    &(200,20)&0.044&0.004&0.003&0.015&0.014&0.023&0.922&0.933&0.844\cr
    &(1000,10)&0.075&0.013&0.003&0.014&0.006&0.010&0.716&0.924&0.841\cr
    &(1000,20)&0.041&0.013&0.002&0.005&0.003&0.004&0.798&0.921&0.867\cr
    \bottomrule
$\tau =0.90$   & (100,10)&-0.059&0.006&-0.001&0.116&0.102&0.186&0.866&0.888&0.775\cr
   & (100,20)&-0.032&-0.002&-0.006&0.049&0.044&0.077&0.897&0.915&0.831\cr
   & (200,10)&-0.058&-0.011&-0.002&0.072&0.047&0.079&0.845&0.920&0.838\cr
   & (200,20)&-0.029&0.002&0.000&0.028&0.023&0.036&0.887&0.929&0.846\cr
   & (1000,10)&-0.058&-0.015&-0.002&0.043&0.011&0.016&0.464&0.892&0.847\cr
   & (1000,20)&-0.030&0.003&-0.000&0.013&0.004&0.007&0.641&0.928&0.839\cr
     \bottomrule
    \end{tabular}
    \begin{tablenotes}
        \footnotesize
        \item {\bf Note:} 1000 replications. DGP: $Y_{it}=(\epsilon_{it}-1)+\epsilon_{it}X_{it}+\alpha_i$, $\alpha_i=\gamma(X_{i1}+\dots+X_{iT}+\lambda_i)-\Ex(\alpha_i)$, $X_{it}\sim i.i.d \text{ }Beta(1,1)$, $\lambda_i\sim i.i.d \text{ }N(0,1)$, $\gamma=2$, $\epsilon_{it}\sim i.i.d \text{ }N(2,1)$.
      \end{tablenotes}
    \end{threeparttable}
\end{center}
}

{\small
 \begin{center}
  \begin{threeparttable}
  \caption{Biases, MSEs and Coverage Rates for Model 2}

    \begin{tabular}{ccccccccccc}
    \toprule
    &  &
    \multicolumn{3}{c}{Biases}&\multicolumn{3}{c}{MSEs}&\multicolumn{3}{c}{Coverage Rates ($95\%$)}\cr
    \cmidrule(lr){3-5} \cmidrule(lr){6-8}  \cmidrule(lr){9-11}
    & $(N,T)$ &$\tilde{\beta}(\tau)$&$\hat{\beta}_{abc}(\tau)$&$\hat{\beta}_{spj}(\tau)$
    &$\tilde{\beta}(\tau)$&$\hat{\beta}_{abc}(\tau)$&$\hat{\beta}_{spj}(\tau)$
    &$\tilde{\beta}(\tau)$&$\hat{\beta}_{abc}(\tau)$&$\hat{\beta}_{spj}(\tau)$\cr
    \midrule
  $\tau =0.25$  &(100,10)&0.048&0.012&-0.003&0.028&0.021&0.041&0.889&0.923&0.806\cr
    &(100,20)&0.021&0.009&0.016&0.008&0.007&0.013&0.906&0.923&0.801\cr
    &(200,10)&0.051&0.003&0.003&0.022&0.011&0.021&0.781&0.931&0.781\cr
    &(200,20)&0.021&-0.001&-0.012&0.005&0.003&0.007&0.872&0.947&0.796\cr
    &(1000,10)&0.048&0.005&-0.005&0.013&0.002&0.004&0.245&0.934&0.792\cr
    &(1000,20)&0.021&-0.001&-0.012&0.002&0.001&0.002&0.526&0.936&0.693\cr
    \bottomrule
 $\tau =0.90$    &(100,10)&-0.048&-0.006&-0.011&0.236&0.242&0.372&0.900&0.905&0.833\cr
    &(100,20)&-0.020&0.008&0.002&0.121&0.132&0.188&0.918&0.920&0.862\cr
    &(200,10)&-0.043&0.018&-0.007&0.135&0.128&0.182&0.896&0.901&0.845\cr
    &(200,20)&-0.021&0.010&-0.002&0.063&0.064&0.088&0.918&0.924&0.855\cr
    &(1000,10)&-0.042&0.011&-0.003&0.053&0.027&0.038&0.744&0.910&0.834\cr
    &(1000,20)&-0.021&0.001&-0.000&0.019&0.012&0.017&0.852&0.924&0.872\cr
     \bottomrule
    \end{tabular}
    \begin{tablenotes}
        \footnotesize
        \item {\bf Note:} 1000 replications. DGP: $Y_{it}=(\epsilon_{it}-1)+\epsilon_{it}X_{it}+\alpha_i$, $\alpha_i=\gamma(X_{i1}+\dots+X_{iT}+\lambda_i)-\Ex(\alpha_i)$, $X_{it}\sim i.i.d \text{ }Beta(1,1)$, $\lambda_i\sim i.i.d \text{ }N(0,1)$, $\gamma=2$, $\epsilon_{it}\sim i.i.d \text{ }\exp(1)+2$.
      \end{tablenotes}
    \end{threeparttable}
\end{center}
}

{\small
 \begin{center}
  \begin{threeparttable}
  \caption{Biases, MSEs and Coverage Rates for Model 3}

    \begin{tabular}{cccccccccccc}
    \toprule
    &  &
    \multicolumn{3}{c}{Biases}&\multicolumn{3}{c}{MSEs}&\multicolumn{3}{c}{Coverage Rates ($95\%$)}\cr
    \cmidrule(lr){3-5} \cmidrule(lr){6-8}  \cmidrule(lr){9-11}
   & $(N,T)$ &$\tilde{\beta}(\tau)$&$\hat{\beta}_{abc}(\tau)$&$\hat{\beta}_{spj}(\tau)$
    &$\tilde{\beta}(\tau)$&$\hat{\beta}_{abc}(\tau)$&$\hat{\beta}_{spj}(\tau)$
    &$\tilde{\beta}(\tau)$&$\hat{\beta}_{abc}(\tau)$&$\hat{\beta}_{spj}(\tau)$\cr
    \midrule
   $\tau =0.25$ &(100,10)&0.056&0.115&-0.123&0.192&0.251&0.345&0.953&0.926&0.847\cr
    &(100,20)&0.004&-0.036&-0.027&0.077&0.085&0.088&0.950&0.933&0.917\cr
    &(200,10)&0.035&-0.009&-0.111&0.089&0.106&0.166&0.957&0.935&0.833\cr
    &(200,20)&-0.007&0.031&-0.021&0.036&0.041&0.042&0.948&0.932&0.911\cr
     &(1000,10)&0.021&0.004&-0.105&0.018&0.021&0.046&0.941&0.926&0.734\cr
    &(1000,20)&-0.001&-0.003&-0.012&0.007&0.008&0.008&0.931&0.917&0.913\cr
    \bottomrule
   $\tau =0.90$ & (100,10)&-0.057&0.002&-0.014&0.057&0.029&0.071&0.830&0.940&0.786\cr
    &(100,20)&-0.034&-0.006&-0.002&0.020&0.009&0.021&0.776&0.936&0.775\cr
    &(200,10)&-0.056&-0.022&-0.015&0.045&0.019&0.033&0.671&0.906&0.800\cr
    &(200,20)&-0.032&0.001&-0.001&0.015&0.005&0.010&0.622&0.936&0.782\cr
     &(1000,10)&-0.058&-0.022&-0.016&0.039&0.008&0.010&0.027&0.724&0.698\cr
    &(1000,20)&-0.033&-0.005&-0.001&0.013&0.001&0.002&0.028&0.902&0.775\cr
     \bottomrule
    \end{tabular}
    \begin{tablenotes}
        \footnotesize
        \item {\bf Note:} 1000 replications. DGP: $Y_{it}=(\epsilon_{it}-1)+\epsilon_{it}X_{it}+\alpha_i$, $\alpha_i=\gamma(X_{i1}+\dots+X_{iT}+\lambda_i)-\Ex(\alpha_i)$, $X_{it}\sim i.i.d \text{ }Beta(1,1)$, $\lambda_i\sim i.i.d \text{ }N(0,1)$, $\gamma=2$, $\epsilon_{it} \sim  i.i.d \text{ }B_{it}\cdot \mathcal{N}(1,1) +(1-B_{it})\cdot\mathcal{N}(3,1)$ with $B_{it} \sim Bernoulli(0.3)$.
      \end{tablenotes}
    \end{threeparttable}
\end{center}
}

{\small
 \begin{center}
  \begin{threeparttable}
  \caption{Biases, MSEs and Coverage Rates for Model 4}

    \begin{tabular}{cccccccccccc}
    \toprule
    &  &
    \multicolumn{3}{c}{Biases}&\multicolumn{3}{c}{MSEs}&\multicolumn{3}{c}{Coverage Rates ($95\%$)}\cr
    \cmidrule(lr){3-5} \cmidrule(lr){6-8}  \cmidrule(lr){9-11}
   & $(N,T)$ &$\tilde{\beta}(\tau)$&$\hat{\beta}_{abc}(\tau)$&$\hat{\beta}_{spj}(\tau)$
    &$\tilde{\beta}(\tau)$&$\hat{\beta}_{abc}(\tau)$&$\hat{\beta}_{spj}(\tau)$
    &$\tilde{\beta}(\tau)$&$\hat{\beta}_{abc}(\tau)$&$\hat{\beta}_{spj}(\tau)$\cr
    \midrule
   $\tau =0.25$ &(100,10)&-0.166&-0.009&-0.005&0.083&0.086&0.143&0.943&0.936&0.863\cr
    &(100,20)&-0.090&0.058&-0.088&0.038&0.040&0.058&0.937&0.942&0.869\cr
    &(200,10)&-0.180&-0.089&-0.050&0.050&0.045&0.065&0.912&0.927&0.862\cr
    &(200,20)&-0.104&0.015&-0.025&0.022&0.019&0.030&0.926&0.943&0.871\cr
     &(1000,10)&-0.175&-0.014&-0.033&0.023&0.008&0.014&0.728&0.944&0.849\cr
    &(1000,20)&-0.093&-0.001&-0.010&0.008&0.004&0.006&0.812&0.941&0.866\cr
    \bottomrule
   $\tau =0.90$ & (100,10)&-0.172&-0.040&-0.024&0.203&0.179&0.296&0.906&0.937&0.852\cr
    &(100,20)&-0.101&-0.011&-0.017&0.093&0.083&0.136&0.917&0.943&0.853\cr
    &(200,10)&-0.165&-0.077&-0.029&0.129&0.102&0.158&0.884&0.923&0.842\cr
    &(200,20)&-0.085&-0.033&-0.004&0.052&0.044&0.067&0.895&0.933&0.864\cr    
     &(1000,10)&-0.169&-0.009&-0.029&0.077&0.019&0.033&0.518&0.948&0.844\cr
    &(1000,20)&-0.088&-0.003&-0.008&0.024&0.008&0.012&0.688&0.944&0.884\cr
     \bottomrule
    \end{tabular}
    \begin{tablenotes}
        \footnotesize
        \item {\bf Note:} 1000 replications. DGP: $Y_{it}=(\epsilon_{it}-1)+\epsilon_{it}X_{it}+\alpha_i$, $\alpha_i=\gamma(X_{i1}+\dots+X_{iT}+\lambda_i)-\Ex(\alpha_i)$, $X_{it}\sim i.i.d \text{ }Beta(1,1)$, $\lambda_i\sim i.i.d \text{ }N(0,1)$, $\gamma=2$, $\epsilon_{it} \sim  i.i.d \text{ }\mathcal{T}(5)$.
      \end{tablenotes}
    \end{threeparttable}
\end{center}
}

There are three main takeaways from the simulation results.

First, compared with the estimator of \cite{canay2011simple}, both analytical bias correction and split-panel jackknife can significantly reduce the bias of our two-step estimator in most cases. The only exception is Model 3 at $\tau=0.25$, where the bias of \cite{canay2011simple}'s estimator is already very small. In particular, in all models, for fixed $T$, the biases of \cite{canay2011simple}'s estimator does not decease as $N$ increases from $100$ to $1000$, confirming the existence of a bias term whose size is determined by $T$ only.

Second, in most cases, the MSE of $\hat{\beta}_{abc}(\tau)$ is the lowest while the MSE of $\hat{\beta}_{spj}(\tau)$ is the highest, implying that although split-panel jackknife is able to reduce the bias it also increases the variance notably.

Third, in many cases, the coverage rates based on $\tilde{\beta}$ is close to the nominal level ($95\%$) when $N$ is not large compared to $T$. However, when $N=1000$ and $T=10,20$, the coverage rates based on $\tilde{\beta}$ is much lower than the nominal level. The most extreme case is in Model 3 at $\tau=0.9$, where the coverage rates based on $\tilde{\beta}$ is less than $3\%$ when $N=1000$. On the other hand, the coverage rates based on $\hat{\beta}_{abc}(\tau)$ is close to the nominal level in almost all cases. It should be noted that the coverage rates based on $\hat{\beta}_{spj}(\tau)$ perform better than $\tilde{\beta}$ when $N$ is large, but in general they are not close to the nominal level due to the relatively high variances of $\hat{\beta}_{spj}(\tau)$.

To sum up, the simulation results above confirm our claim that failing to take into account the asymptotic bias of \cite{canay2011simple}'s estimator will lead to invalid inference especially when $N$ is large and $T$ is small. Our new two-step estimator with analytical bias correction is shown to perform the best in terms of bias correction and the coverage rates of the confidence intervals.

\section{Conclusion}
In this paper, we revisit the popular two-step estimator of \cite{canay2011simple} for quantile panel data models, and explain why the inference based on \cite{canay2011simple}'s estimators of the coefficients and the covariance matrix are not valid. Solving this important problem is crucial for correctly interpreting the empirical findings in more than 80 papers that have employed \cite{canay2011simple}'s estimator.

We propose a new two-step estimator based on smoothed quantile regressions, and establish the asymptotic distribution of the new estimator. In particular, we derive the analytical expression for the asymptotic bias which provides the basis for both analytical and split-panel jackknife bias corrections. In addition, we provide a more accurate characterization of the asymptotic covariance matrix, which is crucial for constructing asymptotically valid confidence intervals. The performance of the new estimator with bias corrections in finite samples is evaluated by Monte Carlo simulations. We find that correcting the asymptotic bias is essential to obtain valid inference in quantile panel data models.

Even though we have provided conditions regarding the size of the bandwidth parameter in the smooth quantile regression, there remains the important question of how to choose the bandwidth parameter optimally in a data-dependent manner. Such an interesting question is left for future research.

%
%
%
%

\appendix

\numberwithin{equation}{section}

\newpage
\section{Proof of the Main Results}

\begin{lemma} Under Assumption 1,\\
(i)\[ \sqrt{NT}( \hat{\theta} -\theta_0) = \mathbf{B}^{-1}\cdot \frac{1}{\sqrt{NT}}\sum_{i=1}^{N}\sum_{t=1}^{T}\ddot{X}_{it}\ddot{\epsilon}_{it}+o_P(1), \]
and \[ \frac{1}{\sqrt{NT}}\sum_{i=1}^{N}\sum_{t=1}^{T}\ddot{X}_{it}\ddot{\epsilon}_{it} =\mathcal{N}(0,\Sigma_{\theta}) +o_P(\sqrt{N}/\sqrt{T}). \]
(ii)
\[ \AverageN |\hat{\alpha}_i -\alpha_{0i} |^2 =O_P\left(\frac{1}{T}\right).\]
\end{lemma}
\begin{proof}
(i) Write
\[ \sqrt{NT}( \hat{\theta} -\theta_0)  = \left(  \frac{1}{NT}\sum_{i=1}^{N}\sum_{t=1}^{T}\ddot{X}_{it}\ddot{X}_{it}'\right)^{-1}\cdot \frac{1}{\sqrt{NT}}\sum_{i=1}^{N}\sum_{t=1}^{T}\ddot{X}_{it}\ddot{\epsilon}_{it} .\]
To save space, we only consider the case where $k=1$. Note that
\[ \Ex[\ddot{X}_{it}^2 ] = \Ex[ X_{it}^2] - 2\Ex[X_{it}\bar{X}_i ]+ \Ex[\bar{X}_i^2 ],\]
and
\[ \Ex[X_{it}\bar{X}_i ] =\Ex[X_{it}\mu_i ] + \Ex[X_{it}(\bar{X}_i-\mu_i) ] =\mu_i^2 + O(T^{-1/2})\]
because
\[|  \Ex[X_{it}(\bar{X}_i-\mu_i) ] | \leq  \frac{1}{\sqrt{T}}\sqrt{\Ex[X_{it}^2]} \cdot \sqrt{  \Ex \left[ \frac{1}{\sqrt{T}}\sum_{t=1}^{T}(X_{it}-\mu_i) \right]^2  }.\]
Similarly, we can show that $\Ex[\bar{X}_i^2 ] = \mu_i^2 + O(T^{-1})$. Thus, it follows that
\[ \Ex[\ddot{X}_{it}^2 ] = \Ex[X_{it}^2 ] -\mu_i^2 +o(1)= \Ex[ (X_{it} -\mu_i)^2]+o(1).\]
Next, define $V_{it} = \ddot{X}_{it}^2 -\Ex[ \ddot{X}_{it}^2]$, we can write
\[\frac{1}{NT}\sum_{i=1}^{N}\sum_{t=1}^{T}\ddot{X}_{it}^2 =\frac{1}{N}\sum_{i=1}^{N}\Ex[ (X_{it} -\mu_i)^2]+\frac{1}{NT}\sum_{i=1}^{N}\sum_{t=1}^{T}V_{it} +o(1).\]
Note that $\Ex[V_{it}]=0$ and that
\[  \Ex\left[\frac{1}{NT}\sum_{i=1}^{N}\sum_{t=1}^{T}V_{it}\right]^2 = \frac{1}{N^2T^2}\sum_{i=1}^{N}\sum_{t=1}^{T}\sum_{s=1}^{T}\Ex[ V_{it}V_{is}] \leq \frac{1}{N^2} \sum_{i=1}^{N}  \left( \frac{1}{T}\sum_{t=1}^{T} \sqrt{ \Ex[ V_{it}^2]} \right)^2=O(N^{-1}), \]
it follows that
\[ \frac{1}{NT}\sum_{i=1}^{N}\sum_{t=1}^{T}\ddot{X}_{it}^2 =\frac{1}{N}\sum_{i=1}^{N}\Ex[ (X_{it} -\mu_i)^2]+o_P(1)=\mathbf{B}+o_P(1).\]
Next,
\[
  \frac{1}{\sqrt{NT}}\sum_{i=1}^{N}\sum_{t=1}^{T}\ddot{X}_{it}\ddot{\epsilon}_{it}
= \frac{1}{\sqrt{NT}}\sum_{i=1}^{N}\sum_{t=1}^{T}\tilde{X}_{it}\epsilon_{it} -  \frac{1}{\sqrt{NT}}\sum_{i=1}^{N}\sum_{t=1}^{T}\tilde{X}_{it}\bar{\epsilon}_{i}.
\]
It is easy to see that $(NT)^{-1/2}\sum_{i=1}^{N}\sum_{t=1}^{T}\tilde{X}_{it}\epsilon_{it} \overset{d}{\rightarrow}\mathcal{N}(0,\Sigma_{\theta})$. Moreover,
\[\frac{1}{\sqrt{NT}}\sum_{i=1}^{N}\sum_{t=1}^{T}\tilde{X}_{it}\bar{\epsilon}_{i} = \frac{\sqrt{N}}{\sqrt{T}} \cdot \frac{1}{N}\sum_{i=1}^{N}
\left( \frac{1}{\sqrt{T}} \sum_{t=1}^{T} \tilde{X}_{it }\right)\cdot \left( \frac{1}{\sqrt{T}} \sum_{t=1}^{T}\epsilon_{it}\right).
 \]
It can be shown that
\[ \Ex\left[\left( \frac{1}{\sqrt{T}} \sum_{t=1}^{T} \tilde{X}_{it }\right)\cdot \left( \frac{1}{\sqrt{T}} \sum_{t=1}^{T}\epsilon_{it}\right)\right]=0 \]
and
\[\Ex  \left[\left( \frac{1}{\sqrt{T}} \sum_{t=1}^{T} \tilde{X}_{it }\right)\cdot \left( \frac{1}{\sqrt{T}} \sum_{t=1}^{T}\epsilon_{it}\right)\right]^2=o(1) ,\]
thus,
\[\frac{1}{\sqrt{NT}}\sum_{i=1}^{N}\sum_{t=1}^{T}\tilde{X}_{it}\bar{\epsilon}_{i} = \frac{\sqrt{N}}{\sqrt{T}} \cdot o_P(1) = o_P(\sqrt{N}/\sqrt{T})\]
and the desired result follows.

(ii) Since $\alpha_{0i} = \bar{Y}_i - \theta_0'\bar{X}_i - \bar{\epsilon}_i$, we have $\hat{\alpha}_i -\alpha_{0i} = (\hat{\theta} -\theta_0)'\bar{X}_i + \bar{\epsilon}_i $, and thus
\[ \AverageN  |\hat{\alpha}_i -\alpha_{0i} |^2 \lesssim  \| \hat{\theta} -\theta_0\|^2 \cdot \AverageN \| \bar{X}_i\|^2 +\AverageN | \bar{\epsilon}_i|^2 \]
First, Assumption 2 implies that
\[ \Ex \| \bar{X}_i\|^2=\Ex \left\| \frac{1}{T}\sum_{t=1}^{T}X_{it} \right\|^2 \leq  \frac{1}{T}\sum_{t=1}^{T} \Ex\|X_{it}\|^2 \leq C \]
for all $i\leq N$. Thus, $N^{-1}\sum_{i=1}^{N} \| \bar{X}_i\|^2 =O_P(1)$.

Second, Assumption 2 implies that
\[ \Ex  | \sqrt{T} \bar{\epsilon}_i |^2  =  \Ex \left| \frac{1}{\sqrt{T}}\sum_{t=1}^{T}\epsilon_{it} \right|^2  \leq C\]
for all $i\leq N$, it then follows that $ N^{-1}\sum_{i=1}^{N} | \bar{\epsilon}_i|^2 = O_P(T^{-1})$. Then the desired result follows because the first result of this Lemma implies that
\[ \hat{\theta} -\theta_0 = O_P\left( \frac{1}{\sqrt{NT}}\right)+o_P\left(\frac{1}{T}\right).\]
\end{proof}

\noindent{\textbf{Proof of Theorem 1:}}
\begin{proof} To simplify the notations, write $\beta_0=\beta(\tau), \hat{\beta} =\hat{\beta}(\tau)$.
Define
\[\mathbb{M}_{NT}(\beta) =\AverageNT\rho_{\tau}(Y_{it} -\beta'W_{it}-\alpha_{0i}), \quad
\mathbb{S}_{NT}(\beta ) =\AverageNT \varrho_{\tau}(Y_{it} -\beta'W_{it}-\alpha_{0i}),
 \]
 \[ \bar{\mathbb{M}}_{NT}(\beta) =\AverageNT\Ex[  \rho_{\tau}(Y_{it} -\beta'W_{it}-\alpha_{0i})], \quad
\hat{\mathbb{S}}_{NT}(\beta ) =\AverageNT  \varrho_{\tau}(Y_{it} -\beta'W_{it} - \hat{\alpha}_{0i}), \]
\[ \mathbb{W}_{NT}(\beta ) = \mathbb{M}_{NT}(\beta ) -\mathbb{M}_{NT}(\beta_0 ) - [ \bar{\mathbb{M}}_{NT}(\beta ) - \bar{\mathbb{M}}_{NT}(\beta_0 )] ,\]
 where $\varrho_{\tau}(u) = [\tau - K(u/h)]u.$

 First, for sufficiently small $\delta>0$, let $B(\delta) =\{ b\in\mathbb{R}^d: \|b-\beta_0\|\leq \delta\}$ be a neighbourhood of $\beta_0$. For any $\bar{\beta} \in B^C(\delta)$, define $r= \delta/\|\bar{\beta}-\beta_0\|$, then the point $\beta^{
\ast} =r \bar{\beta} + (1-r) \beta_0$ is on the boundary of $B(\delta)$ because $\|\beta^{\ast} -\beta_0 \| =r \| \bar{\beta} -\beta_0\| =\delta$. By convexity of the check function, we have
\begin{equation*} r \rho_{\tau}(Y_{it} -\bar{\beta} 'W_{it}-\alpha_{0i})+(1-r)  \rho_{\tau}(Y_{it} -\beta_0 'W_{it}-\alpha_{0i}) \geq \rho_{\tau}(Y_{it} -\beta^{\ast'} W_{it}-\alpha_{0i}),\end{equation*}
or
\begin{equation}\label{A1} r \left[ \rho_{\tau}(Y_{it} -\bar{\beta} 'W_{it}-\alpha_{0i})-\rho_{\tau}(Y_{it} -\beta_0 'W_{it}-\alpha_{0i})  \right]
\geq \rho_{\tau}(Y_{it} -\beta^{\ast'} W_{it}-\alpha_{0i}) - \rho_{\tau}(Y_{it} -\beta_0 'W_{it}-\alpha_{0i}).
\end{equation}

Second, Assumption 1 implies that for some $\underline{c}>0$,
\begin{equation}\label{A2}
 \Ex[\rho_{\tau}(Y_{it} -\beta^{\ast'} W_{it}-\alpha_{0i})] - \Ex[ \rho_{\tau}(Y_{it} -\beta_0 'W_{it}-\alpha_{0i}) ] \geq  \underline{c}\|\beta^{\ast'} - \beta_0 \|^2 = \underline{c} \cdot \delta^2.
 \end{equation}

Third, by definition $ \hat{\mathbb{S}}_{NT}(\hat{\beta} ) \leq  \hat{\mathbb{S}}_{NT}(\beta_0 )$, and adding and subtracting terms give
\begin{multline}\label{A3} \mathbb{M}_{NT}(\hat{\beta} ) -  \mathbb{M}_{NT}(\beta_0)     \leq
 \underbrace{ \left(   \mathbb{M}_{NT}(\hat{\beta} )  -  \mathbb{S}_{NT}(\hat{\beta})  \right) }_{I} -  \underbrace{ \left(  \mathbb{M}_{NT}(\beta_0 ) -   \mathbb{S}_{NT}(\beta_0 ) \right) }_{II}  \\
+ \underbrace{ \left(  \mathbb{S}_{NT}(\hat{\beta} )  -  \hat{\mathbb{S}}_{NT}(\hat{\beta}) \right)}_{III }-   \underbrace{ \left(\mathbb{S}_{NT}(\beta_0)  -\hat{ \mathbb{S}}_{NT}(\beta_0)   \right) }_{IV}.
  \end{multline}

Next, suppose that $\| \hat{\beta} -\beta_0 \| >\delta$, it follows from \eqref{A1} and \eqref{A2} that
\begin{equation}\label{A4}
\underline{c}/r \cdot \delta^2  \leq  \mathbb{M}_{NT}(\hat{\beta} ) -  \mathbb{M}_{NT}(\beta_0 )   + \sup_{ \beta \in B(\delta)} \left\| \mathbb{W}_{NT}(\beta )  \right\|.
\end{equation}
It then follows from \eqref{A3} and \eqref{A4} that
\begin{equation}\label{A5}
P[ \| \hat{\beta} -\beta_0 \| >\delta ] \leq P \left[  |I| +|II|+|III|+|IV| +\sup_{ \beta \in B(\delta)} \left\| \mathbb{W}_{NT}(\beta )  \right\|  \geq  \underline{c}/r \cdot \delta^2 \right].
\end{equation}
It is easy to see that $I$ and $II$ are both $O_P(h)$, and that (using the results of Lemma 1)
\[ |III|+|IV| \lesssim  \AverageN  |\hat{\alpha}_i -\alpha_{0i}| \leq \sqrt{ \AverageN  |\hat{\alpha}_i -\alpha_{0i}|^2} =o_P(1),\]
then the desired result follows if
\begin{equation}\label{A6}
P\left[ \sup_{ \beta \in B(\delta)} \left\| \mathbb{W}_{NT}(\beta )  \right\| >c \right] =o(1) \text{ for any }c>0.
\end{equation}
For any $\epsilon>0$, let $\beta^{(1)},\ldots, \beta^{(m)}$ be a maximal set in $B(\delta)$ such that $ \| \beta^{(l)} - \beta^{(k)}\|\geq\epsilon $ for any $l\neq k$ and $1\leq l,k\leq m$. By the compactness of $B(\delta)$, $m$ is finite. For any $\beta \in B(\delta)$, define $\beta^{\ast} =\{ \beta^{(l)}: 1\leq l\leq m, \| \beta-\beta^{\ast}\|\leq \epsilon \}$. Thus, we can write $\mathbb{W}_{NT}(\beta ) = \mathbb{W}_{NT}(\beta^{\ast} ) + \mathbb{W}_{NT}(\beta)  -\mathbb{W}_{NT}(\beta^{\ast})$ and it follows that
\begin{equation}\label{A7}
  \sup_{ \beta \in B(\delta)} \left\| \mathbb{W}_{NT}(\beta)  \right\| \leq
\max_{1\leq l \leq m}  \left\| \mathbb{W}_{NT}(\beta^{(l)} )  \right\| +
\sup_{ \beta \in B(\delta)} \left\|  \mathbb{W}_{NT}(\beta )  - \mathbb{W}_{NT}(\beta^{\ast})  \right\| .
\end{equation}
For the second term on the RHS of \eqref{A7}, it is easy to show that
\[
\sup_{ \beta \in B(\delta)} \left\|  \mathbb{W}_{NT}(\beta)  - \mathbb{W}_{NT}(\beta^{\ast})   \right\|
\lesssim  \epsilon \cdot  \frac{1}{NT}\sum_{i=1}^{N}\sum_{t=1}^{T} \|W_{it}\| \leq  \epsilon \cdot ( C +  o_P(1))
 \]
 because $ | \rho_{\tau}(Y_{it} -\beta_1'W_{it} -\alpha_{0i}) -  \rho_{\tau}(Y_{it} -\beta_2'W_{it} -\alpha_{0i}) | \lesssim \| W_{it}\| \cdot \|\beta_1-\beta_2\|$, and Assumption 1 implies that $\Ex\|W_{it}\| \leq C$ for all $i,t$ and that $(NT)^{-1}\sum_{i=1}^{N}\sum_{t=1}^{T} \left[  \|W_{it}\| -\Ex\|W_{it}\|  \right] =o_P(1)$. Similarly, it can be shown that the first term on the RHS of \eqref{A7} is $o_P(1)$. Thus, \eqref{A6} follows since $\epsilon$ is arbitrary, and this concludes the proof. \end{proof}

 Define
 \[  \varrho^{(1)}(u) = \tau - K(u/h) + k(u/h)u/h,\quad \varrho^{(2)}(u) = 2k(u/h)1/h + k^{(1)}(u/h)u/h^2 \]
\[ \varrho^{(3)}(u) = 3k^{(1)}(u/h)1/h^2 + k^{(2)}(u/h)u/h^3, \quad \varrho^{(4)}(u) = 4k^{(2)}(u/h)1/h^3 + k^{(3)}(u/h)u/h^4.\]

Write $\varrho^{(1)}_{it} =\varrho^{(1)}(u_{it})$, $\varrho^{(2)}_{it} =\varrho^{(2)}(u_{it})$, $\varrho^{(3)}_{it} =\varrho^{(3)}(u_{it})$.
\begin{lemma}Let $\Delta(\alpha_{0i})$ and $\Delta(\beta_0)$ be neighbourhoods of $\alpha_{0i}$ and $\beta_0$. Under Assumptions 1 and 2, we have\\
(i) \[  \max_{1\leq i\leq N} \sup_{\alpha_i \in \Delta(\alpha_{0i}),\beta\in\Delta(\beta_0) } \left| \frac{1}{T}\sum_{t=1}^{T} \varrho^{(2)}(Y_{it} -\beta'W_{it} -\alpha) -\Ex[ \varrho^{(2)}(Y_{it} -\beta'W_{it} -\alpha) ]\right| =o_P\left( \frac{\log N}{ \sqrt{Th}}\right), \]
\[  \max_{1\leq i\leq N} \sup_{\alpha_i \in \Delta(\alpha_{0i}),\beta\in\Delta(\beta_0) } \left| \frac{1}{T}\sum_{t=1}^{T} \varrho^{(3)}(Y_{it} -\beta'W_{it} -\alpha) -\Ex[ \varrho^{(3)}(Y_{it} -\beta'W_{it} -\alpha) ]\right| =o_P\left( \frac{\log N}{ \sqrt{Th^3}}\right) \]
(ii) $\Ex[\varrho^{(1)}_{it}|X_{it} ]=O(h^q)$, $\Ex[\varrho^{(2)}_{it}W_{it} ]=\gamma_i+O(h^q)$, $\Ex[\varrho^{(3)}_{it}W_{it} ]=\eta_i+O(h^{q-1})$, and
\[ \sup_{\alpha_i \in \Delta(\alpha_{0i}),\beta\in\Delta(\beta_0) } \left\| \Ex[\varrho^{(4)}(Y_{it} -\beta'W_{it} -\alpha ) W_{it} ] \right\| =O(1). \]
\end{lemma}
\begin{proof}
The proof is similar to the proof of Lemma B.1 and Lemma B.2 of \cite{galvao2016smoothed} and therefore it is omitted.
\end{proof}

\noindent{\textbf{Proof of Theorem 2:}}
\begin{proof}
The first order condition (FOC) is given by:
\[      \partial \hat{\mathbb{S}}_{NT}(\hat{\beta})/\partial \beta=\AverageNT  \varrho^{(1)}(Y_{it} -\hat{\beta}'W_{it} - \hat{\alpha}_i)W_{it}= 0 .\]
Expanding the FOC around $(\beta_0,\alpha_{01},\ldots,\alpha_{0N})$ gives:
\begin{multline}\label{A8}
\left( \AverageNT  \varrho^{(2)}_{it}W_{it}W_{it}'\right)(\hat{\beta}-\beta_0)= \AverageNT  \varrho^{(1)}_{it}W_{it}
- \AverageNT  \varrho^{(2)}_{it}W_{it}(\hat{\alpha}_i - \alpha_{0i}) \\
+ 0.5 \sum_{j=1}^{d+1}\sum_{l=1}^{d+1} \left[ \left( \AverageNT  \varrho^{(3)}_{it}(\ast)W_{it}W_{it,j}W_{it,l}\right)(\hat{\beta}_j-\beta_{0,j})(\hat{\beta}_l-\beta_{0,l})\right] \\
+ 0.5\left( \AverageNT  \varrho^{(3)}_{it}(\ast) (\hat{\alpha}_i - \alpha_{0i})W_{it}W_{it}'\right)(\hat{\beta}-\beta_0)
+0.5\left( \AverageNT  \varrho^{(3)}_{it}(\ast) (\hat{\alpha}_i - \alpha_{0i})^2W_{it}\right),
\end{multline}
where $\varrho^{(3)}_{it}(\ast) =\varrho^{(3)}(Y_{it} -\beta^{\ast '}W_{it} - \alpha_i^{\ast})$, and $\beta^{\ast }$ lies between $\beta_0$ and $\hat{\beta}$, $\alpha_i^{\ast}$ lies between $\alpha_{0i}$ and $\hat{\alpha}_i$.

\noindent{\textbf{Step 1:}} We can write
\[\AverageNT  \varrho^{(2)}_{it}W_{it}(\hat{\alpha}_i - \alpha_i) = \left( \AverageN  \bar{\gamma}_{i}\bar{X}_i' \right)(\hat{\theta} -\theta_0) +\AverageN  \bar{\gamma}_{i}\bar{\epsilon}_i \]
where $\bar{\gamma}_{i} = T^{-1}\sum_{t=1}^{T} \varrho^{(2)}_{it}W_{it}$. Define $\tilde{\gamma}_{i } = \bar{\gamma}_{i}-\gamma_i$, we have
\begin{multline}\label{A9}
\AverageNT  \varrho^{(2)}_{it}W_{it}(\hat{\alpha}_i - \alpha_i)=\left( \AverageN  \bar{\gamma}_{i}\bar{X}_i' \right)(\hat{\theta} -\theta_0) + \AverageNT  \gamma_i\epsilon_{it} + \AverageN  \tilde{\gamma}_{i}\bar{\epsilon}_{i}.
\end{multline}
From Lemma 1 and 2 we have
\[\left( \AverageN  \bar{\gamma}_{i}\bar{X}_i' \right)(\hat{\theta} -\theta_0) =\left( \AverageN  \gamma_{i}\mu_i' \right)(\hat{\theta} -\theta_0) +o_P((NT)^{-1/2})=\mathbf{A} (\hat{\theta} -\theta_0)+o_P((NT)^{-1/2}). \]
Next, the last term on the RHS of \eqref{A9} can be written as
\[ \frac{1}{T} \cdot \AverageN \left( \frac{1}{\sqrt{T}}\sum_{t=1}^{T} (   \varrho^{(2)}_{it}W_{it} - \gamma_i  ) \right) \left(\frac{1}{\sqrt{T}}\sum_{t=1}^{T}\epsilon_{it} \right). \]
It can be shown that
\[ \Ex\left[ \AverageN \left( \frac{1}{\sqrt{T}}\sum_{t=1}^{T} (   \varrho^{(2)}_{it}W_{it} - \gamma_i  ) \right) \left(\frac{1}{\sqrt{T}}\sum_{t=1}^{T}\epsilon_{it} \right)      \right]
= \AverageN \Ex[ \varrho^{(2)}_{it}W_{it} \epsilon_{it}].
\]
Note that since $\epsilon_{it} = u_{it} + [\beta(\tau) -\lambda_0]'W_{it}$,
\[\Ex[ \varrho^{(2)}_{it}W_{it} \epsilon_{it}] =\Ex[ \varrho^{(2)}_{it}W_{it} u_{it} ] + \Ex[ \varrho^{(2)}_{it}W_{it}W_{it}'] (\beta_0-\theta_0).\]
Similar to the proof of Lemma 2, we can show that
\[\Ex[ \varrho^{(2)}_{it}W_{it} u_{it} ]  =o(1) \quad \text{ and } \quad \Ex[ \varrho^{(2)}_{it}W_{it}W_{it}'] = \Ex \left[ \fx_{it}(0|X_{it}) W_{it}W_{it}' \right] +O(h^q),\]
it then follows that
\[ \AverageN \Ex[ \varrho^{(2)}_{it}W_{it} \epsilon_{it}]  =\Sigma \cdot (\beta_0-\lambda_0) +o(1).\]
Further, it can be shown that
\[ \Ex\left\| \AverageN \left( \frac{1}{\sqrt{T}}\sum_{t=1}^{T} (   \varrho^{(2)}_{it}W_{it} - \gamma_i  ) \right) \left(\frac{1}{\sqrt{T}}\sum_{t=1}^{T}\epsilon_{it} \right)      \right\|^2 =o(1), \]
thus we have
\[\AverageN  \tilde{\gamma}_{i}\bar{\epsilon}_{i} = \frac{\Sigma \cdot (\beta_0-\lambda_0)}{T}+o_P(T^{-1}). \]
Combining all the above results gives:
\begin{equation}\label{A10}
\AverageNT  \varrho^{(2)}_{it}W_{it}(\hat{\alpha}_i - \alpha_i) = \mathbf{A} (\hat{\theta} -\lambda_0)+\AverageNT  \gamma_i\epsilon_{it}  + \frac{\Sigma \cdot (\beta_0-\theta_0)}{T}+o_P(T^{-1}).
\end{equation}

\noindent{\textbf{Step 2:}} Write
\[ \AverageNT  \varrho^{(3)}_{it}(\ast) (\hat{\alpha}_i - \alpha_i)^2W_{it}
= \AverageN  \left[     \left( \AverageT \varrho^{(3)}_{it}(\ast)W_{it} \right)(\hat{\alpha}_i - \alpha_i)^2  \right].
\]
By Lemma 2 and Assumption 2, we can show that
\[ \AverageT \varrho^{(3)}_{it}(\ast)W_{it} =\eta_i + \bar{o}_P\left( \frac{\log N}{\sqrt{Th^3}}\right) + \bar{O}(1) \left( |\hat{\alpha}_i - \alpha_i | +\|\hat{\beta} -\beta_0\|\right).\]
It can be show that $\max_{1\leq i \leq N}|\hat{\alpha}_i -\alpha_i |=o_P(1)$. Thus, it follows from Lemma 1 and Assumption 2 that:
\[ \AverageNT  \varrho^{(3)}_{it}(\ast) (\hat{\alpha}_i - \alpha_i)^2W_{it} = \AverageN  \eta_i (\hat{\alpha}_i - \alpha_i)^2+o_P(T^{-1}).\]
Finally, it is easy to show that
\[ \AverageN  \eta_i (\hat{\alpha}_i - \alpha_i)^2 =o_P(\| \hat{\theta}-\theta_0\|) + \AverageN \eta_i \bar{\epsilon_i}^2,\]
and that
\[ \AverageN \eta_i \bar{\epsilon_i}^2 = \frac{1}{T}\cdot\AverageN \eta_i \left(\frac{1}{\sqrt{T}}\sum_{t=1}^{T}\epsilon_{it} \right)^2
=\frac{1}{T}  \cdot \AverageN \eta_i \Ex[\epsilon_{it}^2] +o_P(T^{-1}).
\]
Combining the above results gives:
\begin{equation}\label{A11}
\AverageNT  \varrho^{(3)}_{it}(\ast) (\hat{\alpha}_i - \alpha_i)^2W_{it} = o_P(T^{-1}) + \frac{d}{T}.
\end{equation}
where $d= \lim_{N\rightarrow\infty} N^{-1}\sum_{i=1}^{N}\eta_i \Ex[\epsilon_{it}^2]$.

\noindent{\textbf{Step 3:}} It is easy to show that the third and the fourth terms on the RHS of \eqref{A8} are both $o_P(\|\hat{\beta} -\beta_0\|)$, thus it follows from \eqref{A8}, \eqref{A10} and \eqref{A11} that
\begin{multline}\label{A12}
\left( \AverageNT  \varrho^{(2)}_{it}W_{it}W_{it}'\right)(\hat{\beta}-\beta_0) = \AverageNT  \left[\varrho^{(1)}_{it}W_{it}-\gamma_i\epsilon_{it} \right]  -  \mathbf{A} (\hat{\theta} -\theta_0) \\ - \frac{\Sigma \cdot (\beta_0-\lambda_0)}{T}
+ \frac{0.5d}{T}+o_P(\|\hat{\beta} -\beta_0\|)+o_P(T^{-1}) .
\end{multline}
From Lemma 1 we have,
\[  \mathbf{A}(\hat{\theta} -\theta_0) = \AverageNT \mathbf{A}\mathbf{B}^{-1}\ddot{X}_{it}\ddot{\epsilon}_{it}+o_P(T^{-1}).\]
Thus, we can write
\begin{equation} \label{A13}
\AverageNT  \left[\varrho^{(1)}_{it}W_{it}-\gamma_i\epsilon_{it} \right] -  \mathbf{A} (\hat{\theta} -\theta_0)
=\AverageNT  \left[\varrho^{(1)}_{it}W_{it}-\gamma_i\epsilon_{it}- \mathbf{A}\mathbf{B}^{-1}\ddot{X}_{it}\ddot{\epsilon}_{it} \right] +o_P(T^{-1}).
\end{equation}
Similar to Lemma 2, we can show that
\begin{equation}\label{A14}
 \AverageNT  \varrho^{(2)}_{it}W_{it}W_{it}' =  \AverageN   \Ex \left[ \fx_{it}(0|X_{it}) W_{it}W_{it}' \right]  +O(h^q) = \Sigma+ o_P(1).
\end{equation}
Finally, the desired result follows from
\begin{equation}  \frac{1}{\sqrt{NT}} \sum_{i=1}^{N} \sum_{t=1}^{T}\left[\varrho^{(1)}_{it}W_{it}-\gamma_i\epsilon_{it}-\mathbf{A}\mathbf{B}^{-1} \ddot{X}_{it}\ddot{\epsilon}_{it} \right]
\overset{d}{\rightarrow} \mathcal{N}(0,\Omega).\end{equation}
From the proof of Lemma 1 we have $(NT)^{-1/2}  \sum_{i=1}^{N} \sum_{t=1}^{T} \ddot{X}_{it}\ddot{\epsilon}_{it}=(NT)^{-1/2}  \sum_{i=1}^{N} \sum_{t=1}^{T} \tilde{X}_{it}\epsilon_{it} +o_P(1)$. Thus, we have
\[ \frac{1}{\sqrt{NT}} \sum_{i=1}^{N} \sum_{t=1}^{T}\left[\varrho^{(1)}_{it}W_{it}-\gamma_i\epsilon_{it}-\mathbf{A}\mathbf{B}^{-1} \ddot{X}_{it}\ddot{\epsilon}_{it} \right] =\frac{1}{\sqrt{NT}} \sum_{i=1}^{N} \sum_{t=1}^{T}\left[\varrho^{(1)}_{it}W_{it}-\gamma_i\epsilon_{it}-\mathbf{A}\mathbf{B}^{-1} \tilde{X}_{it}\epsilon_{it} \right]+o_P(1). \]
Note that by Lemma 2,
\[ \Ex\left[\varrho^{(1)}_{it}W_{it}-\gamma_i\epsilon_{it}-\mathbf{A}\mathbf{B}^{-1} \tilde{X}_{it}\epsilon_{it} \right] =O(h^q).\]
Moreover,
\begin{align*}
&  \Ex\left[ \left(\varrho^{(1)}_{it}W_{it}-\gamma_i\epsilon_{it}-\mathbf{A}\mathbf{B}^{-1} \tilde{X}_{it}\epsilon_{it} \right)\left(\varrho^{(1)}_{it}W_{it}-\gamma_i\epsilon_{it}-\mathbf{A}\mathbf{B}^{-1} \tilde{X}_{it}\epsilon_{it} \right)'     \right]      \\
=& \Ex\left[  \left(\varrho^{(1)}_{it}\right)^2 W_{it}W_{it}' \right]+ \gamma_i \gamma_i' \cdot \Ex\left[\epsilon_{it}^2 \right]+\mathbf{A}\mathbf{B}^{-1} \Ex\left[ \epsilon_{it}^2 \tilde{X}_{it}\tilde{X}_{it} \right]\mathbf{B}^{-1} \mathbf{A}'-2\Ex\left[ \varrho^{(1)}_{it}W_{it} \epsilon_{it}\right]\gamma_i' \\
&-2\Ex\left[ \varrho^{(1)}_{it}W_{it}\tilde{X}_{it}'\epsilon_{it} \right]\mathbf{B}^{-1} \mathbf{A}'+2 \gamma_i\Ex\left[\tilde{X}_{it}'\epsilon_{it}^2 \right]\mathbf{B}^{-1} \mathbf{A}',
\end{align*}
and similar to Lemma 2 we can show that
\[ \Ex\left[  \left(\varrho^{(1)}_{it}\right)^2 W_{it}W_{it}' \right] =\tau(1-\tau)\cdot \Ex[  W_{it}W_{it}'] +o(1),\]
\[\Ex\left[ \varrho^{(1)}_{it}W_{it} \epsilon_{it}\right]\gamma_i' = \Ex\left[ \left(\tau - \mathbf{1}\{u_{it}\leq 0\}\right)u_{it}W_{it}      \right] \gamma_i'+o(1),\]
\begin{multline*}
\Ex\left[ \varrho^{(1)}_{it}W_{it}\tilde{X}_{it}'\epsilon_{it} \right]\mathbf{B}^{-1} \mathbf{A}'  =\Ex\left[ \left(\tau - \mathbf{1}\{u_{it}\leq 0\}\right)u_{it}W_{it}X_{it}'      \right]\mathbf{B}^{-1} \mathbf{A}'  \\- \Ex\left[ \left(\tau - \mathbf{1}\{u_{it}\leq 0\}\right)u_{it}W_{it}      \right] \mu_i'\mathbf{B}^{-1} \mathbf{A}' +o(1).  \end{multline*}
Thus, we have
\begin{equation*}
 \frac{1}{NT}\sum_{i=1}^{N} \sum_{t=1}^{T} \Ex\left[ \left(\varrho^{(1)}_{it}W_{it}-\gamma_i\epsilon_{it}-\mathbf{A}\mathbf{B}^{-1} \tilde{X}_{it}\epsilon_{it} \right)\left(\varrho^{(1)}_{it}W_{it}-\gamma_i\epsilon_{it}-\mathbf{A}\mathbf{B}^{-1} \tilde{X}_{it}\epsilon_{it} \right)'     \right]  \rightarrow \Omega,
\end{equation*}
and the desired result follows from Lyapunov's central limit theorem.
\end{proof}

\newpage
\begin{spacing}{1.2}
\bibliographystyle{chicago}
\bibliography{inter_reference}
\end{spacing}

\end{document}